\documentclass[a4paper]{article}

\usepackage[english]{babel}
\usepackage{amsmath,amsthm,amssymb,amsfonts}
\usepackage{color}
\usepackage{graphicx}
\usepackage{subcaption}
\usepackage{url}
\usepackage{hyperref}
\usepackage{authblk}

\newtheorem{theorem}{Theorem}
\newtheorem{fact}{Fact}
\newtheorem{lemma}{Lemma}
\newtheorem{definition}{Definition}

\newcommand{\set}[1]{\left\{ #1 \right\}}

\newcommand{\G}{\mathcal{G}}
\newcommand{\U}{\mathcal{U}}
\newcommand{\A}{\mathcal{A}}
\newcommand{\VE}{\text{V}_{\text{Eve}}}
\newcommand{\VA}{\text{V}_{\text{Adam}}}

\newcommand{\Parity}{\text{Parity}}

\begin{document}

\title{Parity games and universal graphs\footnote{T. Colcombet and N. Fijalkow are supported by the DeLTA project (ANR-16-
CE40-0007) and N. Fijalkow is supported by The Alan Turing Institute under the EPSRC grant EP/N510129/1.}}
\author[1]{Thomas Colcombet}
\author[2,3]{Nathana{\"e}l Fijalkow}
\affil[1]{CNRS, IRIF, Universit{\'e} Paris Diderot, France}
\affil[2]{CNRS, LaBRI, Universit{\'e} de Bordeaux, France}
\affil[3]{The Alan Turing Institute of data science, London}
\date{}

\maketitle

\begin{abstract}
This paper is a contribution to the study of parity games and the recent constructions of three quasipolynomial time algorithms
for solving them.
We revisit a result of Czerwi{\'n}ski, Daviaud, Fijalkow, Jurdzi{\'n}ski, Lazi{\'c}, and Parys 
showing a quasipolynomial barrier for all three quasipolynomial time algorithms.
The argument is that all three algorithms can be understood as constructing a so-called separating automaton,
and to give a quasipolynomial lower bound on the size of separating automata.

We give an alternative proof of this result.
The key innovations of this paper are the notion of universal graphs and the idea of saturation.
\end{abstract}

\section{The quasipolynomial era for parity games}

\begin{figure}[!ht]
\centering
\includegraphics[width=.8\linewidth]{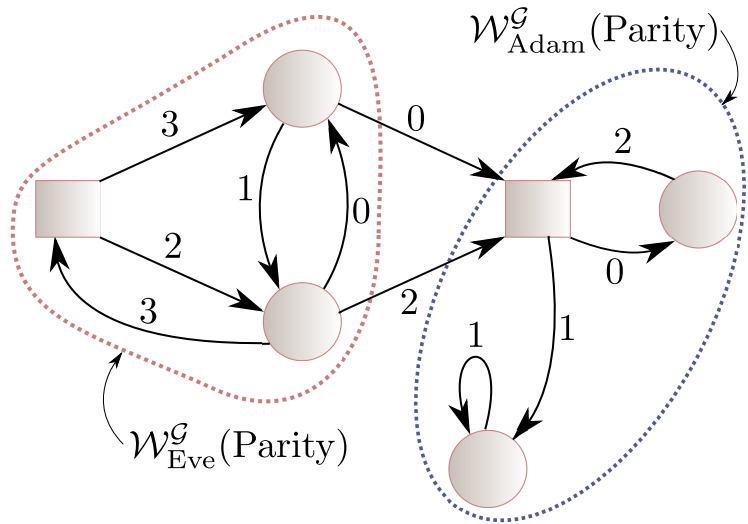}
\label{fig:parity_game}
\caption{A parity game (the dotted regions are the winning sets for each player).}
\end{figure}

The first quasipolynomial time algorithm was constructed by Calude, Jain, Khoussainov, Li, and Stephan~\cite{CJKLS17}.
Shortly after its publication two papers came out to give different presentations and correctness proofs of the algorithm,
the first by Gimbert and Ibsen-Jensen~\cite{GIJ17}, the second by Fearnley, Jain, Schewe, Stephan, and Wojtczak~\cite{FJSSW17}.
Later Boja{\'n}czyk and Czerwi{\'n}ski~\cite{BC18} explained how to understand the algorithm as the construction
of a so-called separating automaton.

The second quasipolynomial time algorithm was defined by Jurdzi{\'n}ski and Lazi{\'c}~\cite{JL17}.
It is called the succinct progress measure algorithm and is presented as an improvement over the previously known small progress measure algorithm.
A year later Fijalkow~\cite{Fij18} introduced the notion of universal trees and argued that the succinct progress measure algorithm is concisely explained using this new notion, and proved a quasipolynomial lower bound for universal trees.

A third quasipolynomial time algorithm was proposed by Lehtinen~\cite{Leh18}, based on the notion of register games.

\vskip1em
Each of the three algorithms can be understood as constructing a (different) separating automaton.
In a recent paper Czerwi{\'n}ski, Daviaud, Fijalkow, Jurdzi{\'n}ski, Lazi{\'c}, and Parys~\cite{CDFJLP18} 
(superseding the technical report~\cite{Fij18}) 
showed that any separating automaton contains in its state space a universal tree, 
which combined with the quasipolynomial lower bound for universal trees 
yields a quasipolynomial barrier applying to the three quasipolynomial time algorithms.

The purpose of this paper is to give an alternative proof of this latest result, stated in Theorem~\ref{thm:main}.
The key innovations are the notion of universal graphs and the idea of saturation.

\begin{theorem}[the equivalence between 1. and 2. was proved in~\cite{CDFJLP18}]\label{thm:main}
The following quantities are equal.
\begin{enumerate}
	\item The size of the smallest universal tree
	\item The size of the smallest separating automaton
	\item The size of the smallest universal graph
\end{enumerate}
\end{theorem}

The construction of a universal tree from a separating automaton is the main technical achievement of 
the paper by Czerwi{\'n}ski, Daviaud, Fijalkow, Jurdzi{\'n}ski, Lazi{\'c}, and Parys~\cite{CDFJLP18}. 
In the present paper we propose a different route: 
we define a saturation technique and show that the saturation of a separating automaton is a universal graph,
and similarly that the saturation of a universal graph is a universal tree.

\section{Universal graphs}

We fix $n,d$ two positive integers parameters.
We let $\Parity \subseteq [0,d-1]^\omega$ denote the set of infinite words such that 
the maximal priority appearing infinitely often is even.

A graph is a structure with $d$ binary relations $E_i$ for $i \in [0,d-1]$, 
with $(v,v') \in E_i$ meaning that there is an edge from $v$ to $v'$ labelled $i$.
We write $(v,i,v') \in E$ instead of $(v,v') \in E_i$.
The size of a graph is its number of vertices.
A graph homomorphism is simply a homomorphism of such structures: 
for two graphs $G,G'$, a homomorphism $\phi : G \to G'$
maps the vertices of $G$ to the vertices of $G'$ such that
\[
(v,i,v') \in E \ \implies\ (\phi(v),i,\phi(v')) \in E
\]
As a simple example that will be useful later on, note that if $G'$ is a super graph of $G$, 
meaning they have the same domain and every edge in $G$ is also in $G'$, then the identity is a homomorphism $G \to G'$.

A graph together with two sets $\VE$ and $\VA$ such that $V = \VE \uplus \VA$ is called a (parity) game:
$\VE$ is the set of vertices controlled by Eve and represented by circles, and $\VA$ the set of vertices controlled by Adam
and represented by squares.
We speak of an $(n,d)$-graph or $(n,d)$-game when it has at most $n$ vertices and $d$ priorities.

\vskip1em
A path is a sequence of triples $(v,i,v')$ in $E$ such that the third component of a triple in the sequence matches 
the first component of the next triple. 
(As a special case we also have empty paths consisting of only one vertex.)
For a path $\rho$ we write $\pi(\rho)$ for its projection over the priorities, meaning
the induced sequence of priorities.
We say that a graph satisfies parity if all paths in the graph satisfy the parity objective. 
Note that this is equivalent to asking whether all cycles are even, meaning the maximal priority appearing in the cycle is even.

\begin{definition}
A graph is $(n,d)$-\textit{universal} if it satisfies parity and 
any $(n,d)$-graph satisfying parity can be mapped homomorphically into it.
\end{definition}

We show that universal graphs can be used to construct a conceptually simple algorithm for parity games.
Consider a $(n,d)$-game $\G$ and a $(n,d)$-universal graph $\U$, we construct a safety game $\G \times \U$ 
where Eve chooses which edge to follow on the second component,
and she wins if she manages to play forever.

\begin{lemma}
Let $\G$ be a $(n,d)$-game and $\U$ a $(n,d)$-universal graph.
Then Eve has a strategy in $\G$ ensuring $\Parity$ if and only if she has a strategy ensuring to play forever in $\G \times \U$.
\end{lemma}
This lemma can be used to algorithmically reduce parity games to safety games,
yielding an algorithm whose complexity is proportional to the size of $\U$.

\begin{proof}
Let us assume that Eve has a strategy $\sigma$ in $\G$ ensuring $\Parity$, which can be chosen to be positional.
We consider the graph $\G[\sigma]$, by definition there exists a homomorphism from $\G[\sigma]$ to $\U$. 
We construct a winning strategy in $\G \times \U$ by playing as in $\sigma$ in the first component
and following the homomorphism on the second component.
Conversely, a strategy in $\G \times \U$ ensuring to play forever in $\G \times \U$ 
induces a strategy in $\G$ ensuring $\Parity$ since $\U$ satisfies parity.
\end{proof}

\section{Universal trees and separating automata}

\subsection*{Universal trees}
The trees we consider have three properties: they are rooted, every leaf has the same height, and the children of a node are totally ordered. 
We say that a tree $t$ embeds into another tree $T$ if $t$ is a subtree of $T$, or equivalently $t$ can be obtained by removing
some nodes of $T$.
The size of a tree is its number of leaves.
A $(n,d)$-tree is a tree with $n$ leaves each at height $\frac{d}{2}$.

\begin{definition}
A tree is $(n,d)$-\textit{universal} if it embeds all $(n,d)$-trees.
\end{definition}

\begin{figure*}[!ht]
\centering
\begin{subfigure}{.5\textwidth}
  \centering
  \includegraphics[width=.8\linewidth]{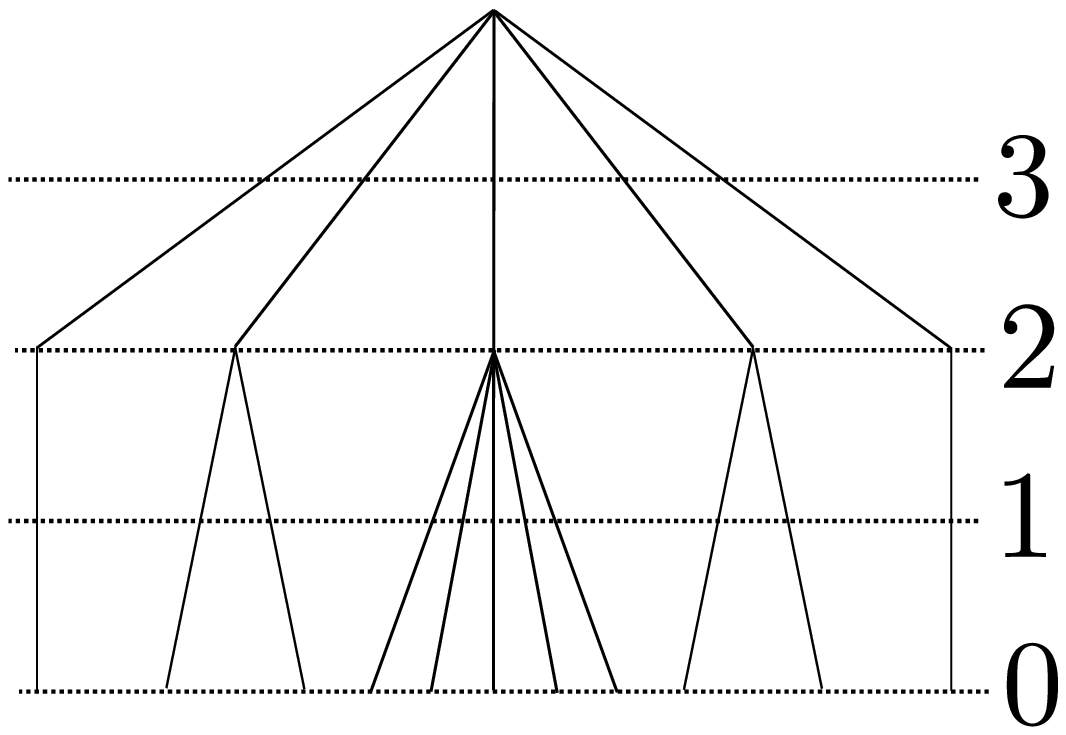}
  \label{fig:example_universal}
\end{subfigure}%
\begin{subfigure}{.5\textwidth}
  \centering
  \includegraphics[width=.8\linewidth]{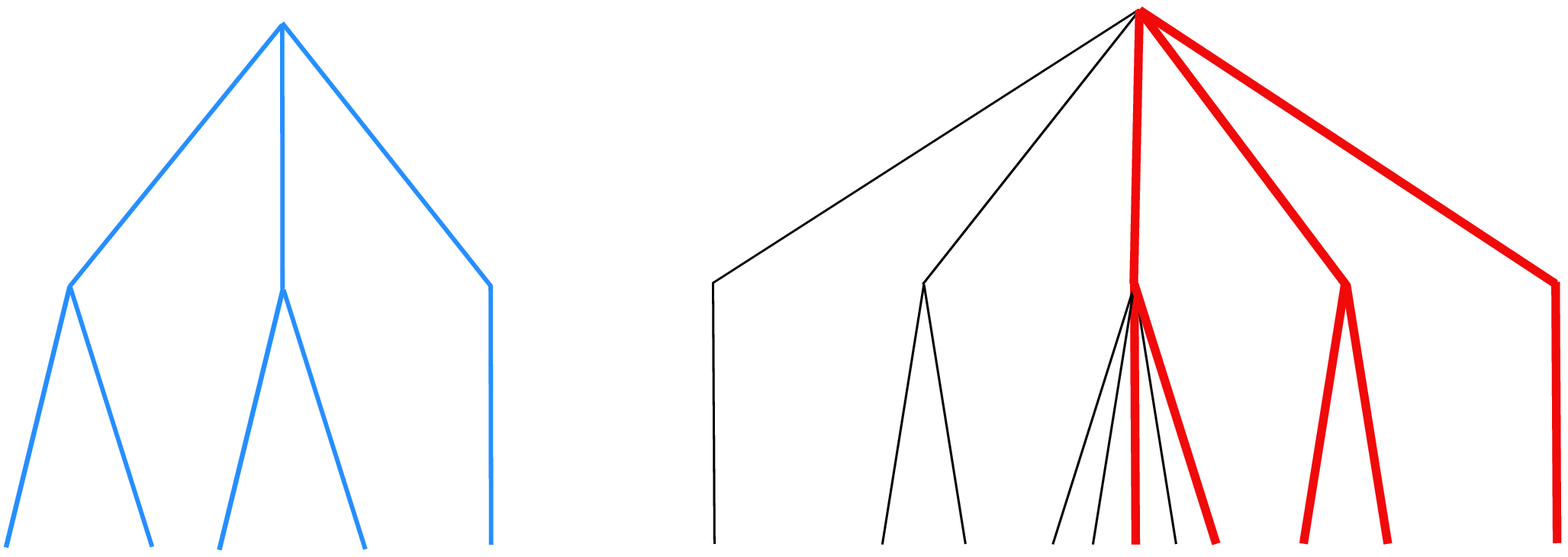}
  \label{fig:example_embedding}
\end{subfigure}
\caption{On the left, a tree for $d = 4$.
This tree is $(5,4)$-universal, and is actually the smallest having this property.
On the right, a tree of size $5$ and one possible embedding into the universal tree.}
\end{figure*}

On a tree the levels are labelled by priorities from bottom to top.
More precisely, even priorities sit on levels corresponding to nodes, 
and odd priorities inbetween levels.

\subsection*{From trees to graphs to trees}

We first explain how a tree induces a graph.
Let $T$ be a tree, we construct a graph $G(T)$ whose vertices are the leaves of $T$
and such that $(v,i,v') \in E$ if
\begin{itemize}
	\item for $i$ even, the ancestor of $v$ at level $i$ is to the left of or equal to the ancestor of $v'$ at level $i$
	\item for $i$ odd, the edge ancestor of $v$ at level $i$ is strictly to the left of the edge ancestor of $v'$ at level $i$
\end{itemize}
Equivalently for $i$ odd, $(v,i,v') \in E$ if and only if $(v,i-1,v') \in E$ and $(v',i-1,v) \notin E$.

\begin{figure}[!ht]
\centering
\includegraphics[width=.8\linewidth]{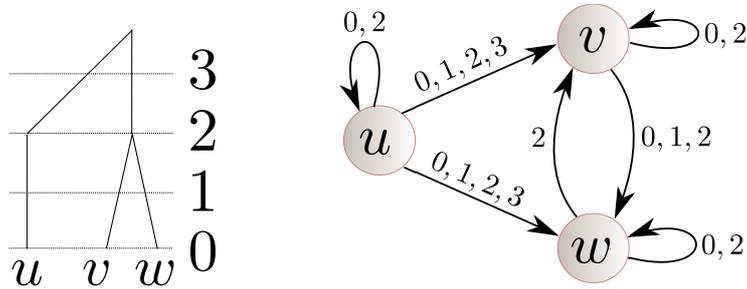}
\label{fig:example_tree_to_graph}
\caption{On the left, a tree $T$, and on the right, the graph $G(T)$.}
\end{figure}

\begin{fact}
For $t,T$ two trees, the following are equivalent.
\begin{itemize}
	\item $t$ embeds in $T$
	\item there exists a homomorphism $\phi : G(t) \to G(T)$
\end{itemize} 
\end{fact}

We define the converse transformation.
We say that a graph satisfying the following four properties is tree-like.
\begin{itemize}
	\item if $(v,i,v') \in E$ and $(v',j,v'') \in E$, then $(v,\max(i,j),v'') \in E$
	\item for $i$ even, $E_i$ is total: either $(v,i,v') \in E$ or $(v',i,v) \in E$ (possibly both)
	\item for $i$ even, $E_i$ is reflexive: $(v,i,v) \in E$
	\item for $i$ odd, $(v,i,v') \in E$ if and only if $(v,i-1,v') \in E$ and $(v',i-1,v) \notin E$
\end{itemize}
Note that the first item implies that $E_i$ is transitive 
and the second and third items imply that $E_0 \subseteq E_2 \subseteq \cdots \subseteq E_{d-1}$.
Further, for $i$ odd $E_i$ is non-reflexive.

Let $G$ be a tree-like graph, we construct a tree $T(G)$ whose leaves are the vertices of $G$.
For $i$ even, the nodes at level $i$ are the equivalence classes of vertices for the total preorder $E_i$.

\begin{fact}
Let $t$ be a tree, then $G(t)$ is tree-like, and $T(G(t)) = t$.
\end{fact}

In the remainder of this article 
we identify trees and tree-like graphs through the two reciprocal transformations described above.
Hence we see trees as special graphs.

\subsection*{Separating automata}
The automata we consider are non-deterministic safety automata over infinite words on the alphabet $[0,d-1]$, 
where safety means that all states are accepting: a word is rejected if there exist no run for it.

In the following we say that a path in a graph is accepted or rejected by an automaton; this is an abuse of language
since what the automaton reads is only the priorities of the corresponding path.

\begin{definition}
An automaton is $(n,d)$-\textit{separating} if the two following properties hold.
\begin{itemize}
	\item For all $(n,d)$-graphs satisfying parity, the automaton accepts all paths in the graph
	\item The automaton rejects all paths not satisfying parity
\end{itemize}
\end{definition}
We let $\Parity_n$ denote the union over all $(n,d)$-graphs satisfying parity of their set of paths.
\[
\Parity_n = \bigcup \set{ \text{Paths}(G) : G (n,d)-\text{graph satisfying } \Parity}
\]

\begin{figure}[!ht]
\centering
\includegraphics[width=.5\linewidth]{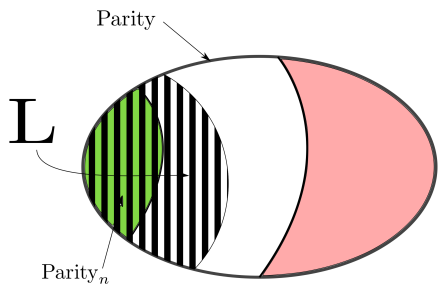}
\label{fig:separating_automata}
\caption{Separating automata.}
\end{figure}

The following lemma justifies the definition of separating automata.

\begin{lemma}
Let $L$ be the language recognised by a $(n,d)$-separating automaton.
Then for all $(n,d)$-games $\G$, Eve has a strategy ensuring $\Parity$ 
if and only if she has a strategy ensuring $L$.
\end{lemma}

\begin{proof}
Let us assume that Eve has a strategy $\sigma$ ensuring $\Parity$, which can be chosen positional.
We construct the $(n,d)$-graph $\G[\sigma]$, by definition it satisfies parity, hence the strategy $\sigma$ also ensures $L$
thanks to the first item of the definition of separating automata.

Conversely, the second item of the definition of separating automata reads $L \subseteq \Parity$,
which implies that any strategy ensuring $L$ also ensures $\Parity$.
\end{proof}

\section{The saturation technique}

The following theorem characterises graphs satisfying parity using homomorphisms into trees.
Since a tree is a tree-like graph, by ``homomorphism from a graph to a tree'' 
we mean homorphisms from a graph to a tree-like graph, 
which are structures over the same signature.

This theorem is an elaboration of the classic result about signature assignments or progress measures being witnesses 
of positional winning strategies in parity games (see \textit{e.g.} Jurdzi{\'n}ski's small progress measure treatment~\cite{J00}).
The novelty here is to phrase this theorem using graph homomorphisms and universal trees,
and giving a different argument using saturation.

\begin{theorem}\label{thm:progress_measure}
Let $G$ be a $(n,d)$-graph and $T$ a $(n,d)$-universal tree. 
The following statements are equivalent
\begin{enumerate}
	\item $G$ satisfies parity 
	\item there exists a $(n,d)$-tree $t$ and a homomorphism $\phi : G \to t$
	\item there exists a homomorphism $\psi : G \to T$
\end{enumerate}
\end{theorem}

For the implication $1 \implies 2$ we introduce maximal graphs satisfying parity.
We say that a graph $G$ is a \textit{maximal graph satisfying parity} 
if adding any edge to $G$ would introduce an odd cycle.

\begin{lemma}\label{lem:maximal_graph}
A maximal graph satisfying parity is tree-like.
\end{lemma}

\begin{proof}
Let $G$ be a maximal graph satisfying parity.
We need to show that $G$ satisfies the following properties.
\begin{itemize}
	\item if $(v,i,v') \in E$ and $(v',j,v'') \in E$, then $(v,\max(i,j),v'') \in E$
	\item for $i$ even, $E_i$ is total: either $(v,i,v') \in E$ or $(v',i,v) \in E$ (possibly both)
	\item for $i$ even, $E_i$ is reflexive: $(v,i,v) \in E$
	\item for $i$ odd, $(v,i,v') \in E$ if and only if $(v,i-1,v') \in E$ and $(v',i-1,v) \notin E$
\end{itemize}

For the first item, if adding $(v,\max(i,j),v'') \in E$ would create an odd cycle, then replacing the edge by the two consecutive edges 
$(v,i,v') \in E$ and $(v',j,v'') \in E$ would yield an odd cycle, contradiction.

For the second item, if we do not have $(v,i,v') \in E$ then there exists a path from $v'$ to $v$ 
with maximal priority odd and larger than $i$,
and similarly for $(v',i,v) \in E$. If both cases would occur, this would induce an odd cycle.

The third and fourth items are clear.
\end{proof}

We can now prove the implication $1 \implies 2$.
Consider a $(n,d)$-graph $G$ satisfying parity, 
we construct a maximal graph $t$ satisfying parity by starting from $G$
and throwing in new edges as long as the graph satisfies parity.
We say that $t$ is a saturation of $G$, it is by construction a super graph of $G$
hence $G$ maps homomorphically into~$t$.
Thanks to Lemma~\ref{lem:maximal_graph}, $t$ is a tree.

\vskip1em
The implication $2 \implies 3$ is obtained by composing homomorphisms: 
indeed, since $T$ is universal there exists a homomorphism $\phi' : t \to T$, which yields
a homomorphism $\phi' \circ \phi : G \to T$.

\vskip1em
The implication $3 \implies 1$ is given by the following lemma, which is a classical argument 
here phrased using trees.

\begin{lemma}
Let $G$ be a graph such that there exists a tree $t$ and a homomorphism $\phi : G \to t$,
then $G$ satisfies parity.
\end{lemma}

\begin{proof}
We consider a loop in $G$.
\[
(v_1, i_1, v_2) (v_2, i_2, v_3) \cdots (v_k, i_k, v_1)
\]
Let us assume towards contradiction that its maximal priority is odd and without loss of generality it is $i_1$.
Applying the homomorphism $\phi$ we have
\[
(\phi(v_1),i_1,\phi(v_2)) \in E,\
(\phi(v_2),i_2,\phi(v_3)) \in E,\ \ldots, \
(\phi(v_k),i_k,\phi(v_1)) \in E
\]
We obtain thanks to the first property of trees that
$(\phi(v_1),i_1,\phi(v_1)) \in E$, which contradicts that for $i$ odd $E_i$ is non-reflexive.
\end{proof}

\section{The three equivalences}

In order to prove Theorem~\ref{thm:main}, we show the following constructions.
\begin{itemize}
	\item From a $(n,d)$-universal tree we construct a (deterministic) $(n,d)$-separating automaton
	\item From a $(n,d)$-universal graph we construct a $(n,d)$-universal tree
	\item From a (non-deterministic) $(n,d)$-separating automaton we construct a $(n,d)$-universal graph
\end{itemize}
In each of these constructions the object constructed has exacly the same size as the original object (for the respective notions of sizes).

\subsection*{From a universal graph to a universal tree}

\begin{lemma}
A maximal $(n,d)$-universal graph is a $(n,d)$-universal tree.
\end{lemma}

This is a direct corollary of Lemma~\ref{lem:maximal_graph}.
It follows that from a $(n,d)$-universal graph, one constructs a $(n,d)$-universal tree by taking a saturation of the universal graph,
which preserves the size.

\subsection*{From a separating automaton to a universal graph}

Let $\A$ be a $(n,d)$-separating automaton. 
We will do the proof assuming that $\A$ is deterministic, and later explain why it straightforwardly extend to non-deterministic automata.
We write $\delta(u)$ for the state reached after reading the word $u$ in $\A$ from the initial state; 
note that $\delta(u)$ might be undefined.
Without loss of generality all states in $\A$ are reachable.
We construct a graph $G$ as follows.
The set of vertices is the set of states of $\A$, and $(v,i,v') \in E$ if $v' = \delta(v,i)$.

\begin{lemma}
The following two properties hold.
\begin{itemize}
	\item $G$ satisfies parity
	\item Any saturation of $G$ is universal
\end{itemize}
\end{lemma}

\begin{proof}
To see that $G$ satisfies parity, we observe that all cycles of $G$ are even, since otherwise $\A$ would accept a word not satisfying parity.
We consider a saturation of $G$ which for the sake of simplicity we also call $G$.
Thanks to Lemma~\ref{lem:maximal_graph} we know that the relation $E_0$ is a total order.

Let $H$ be a $(n,d)$-graph satisfying parity, we construct a homomorphism $\phi$ from $H$ to $G$.
We need to associate to any vertex $v$ of $H$ a vertex $\phi(v)$ of $G$.
Define
\[
\phi(v) = \max_{E_0} \set{\delta(\pi(\rho)) : \rho \text{ path of } H \text{ ending in } v},
\]
where $\pi(\rho)$ projects the path $\rho$ onto a sequence of priorities.
Note that since $H$ satisfies parity, for each path $\rho$ of $H$ we have that $\delta(\pi(\rho))$ is well defined.

We verify that $\phi$ is a homomorphism.
Let $(v,i,v') \in E$, we show that $(\phi(v),i,\phi(v')) \in E$.
By definition $\phi(v) = \delta(\pi(\rho))$ for $\rho$ some path ending in $v$.
Note that $\rho' = \rho \cdot (v,i,v')$ is a path ending in $v'$, so by definition $(\delta(\pi(\rho')),0,\phi(v')) \in E$.
Since $\delta(\pi(\rho')) = \delta(\delta(\pi(\rho)),i)$ and $\phi(v) = \delta(\pi(\rho))$ we have $(\phi(v),i,\delta(\pi(\rho'))) \in E$.
It follows thanks to the first property of trees that $(\phi(v),i,\phi(v')) \in E$.
\end{proof}

The proof goes through if $\A$ is non-deterministic with two adjustments
\begin{itemize}
	\item $\phi(v)$ is the maximum for the $E_0$ order over all states reachable through some path of $H$ ending in $v$
	\item we do not have that $\phi(v) = \delta(\pi(\rho))$, but only that $\phi(v) \in \delta(\pi(\rho))$
\end{itemize}

\subsection*{From a universal tree to a separating automaton}

Let $T$ be a $(n,d)$-universal tree. 
We construct an automaton $\A$ as follows.
The set of states is the set of leaves of $T$. 
The initial state is the minimum element for the total order $E_0$.
The transition function is defined as follows.
\[
\delta(v,i) = \min_{E_0} \set{v' : (v,i,v') \in E}
\]
Note that $\A$ is deterministic.

\begin{lemma}
The automaton $\A$ is $(n,d)$-separating.
\end{lemma}

\begin{proof}
Consider a $(n,d)$-graph $G$ satisfying parity, we show that $\A$ accepts all paths in $G$.
Thanks to Theorem~\ref{thm:progress_measure}, there exists a tree and a homomorphism $\phi : G \to t$.

We show by induction that for a finite path $\rho$ in $G$ ending in $v$, 
we have $(\delta(\pi(\rho)), 0, \phi(v)) \in E$.
To initialise we recall that the initial state is the minimal element for $E_0$.
For the inductive case let $\rho' = \rho \cdot (v,i,v')$, we are looking at $\delta(\pi(\rho')) = \delta(\delta(\pi(\rho)),i)$.
Since $(v,i,v') \in E$ we have $(\phi(v),i,\phi(v')) \in E$,
and by induction hypothesis we have $(\delta(\pi(\rho)), 0, \phi(v)) \in E$,
the combination of this implies that $(\delta(\pi(\rho)), i, \phi(v')) \in E$.
It follows by definition of $\delta$ that $(\delta(\pi(\rho')), 0, \phi(v')) \in E$.

To see that the automaton rejects all paths not satisfying parity, consider a path $\rho$ with odd maximal priority $i$.
We note that for any state $v$ we have $(v,i,\delta(v,\pi(\rho))) \in E$, which is easily shown by induction.
Since a path not satisfying parity contains a suffix consisting of infinitely many paths 
each of which has odd maximal priority $i$, it follows that the automaton eventually rejects such a path.
Indeed, $i$ being odd $E_i$ is a non-reflexive preorder, so the sequence of states at the beginning of each loop 
would be an infinite increasing sequence of states.
\end{proof}

\section*{Acknowledgments}
We would like to thank Marcin Jurdzi{\'n}ski for his detailed comments on the first version of this document.

\bibliographystyle{alpha}
\bibliography{bib}

\end{document}